\documentclass[DIV9,a4paper,final]{scrartcl}
 
\usepackage[utf8]{inputenc}
\usepackage[T1]{fontenc}

\usepackage{microtype}
\usepackage{lmodern}
\usepackage{nccmath}
\usepackage{booktabs}

\usepackage{algorithm}
\usepackage{algpseudocode}

\usepackage{amsmath}
\usepackage{amssymb}
\usepackage{enumitem}
\usepackage{tabularx}

\newtheorem{theorem}{Theorem}
\newtheorem{lemma}[theorem]{Lemma}

\newtheorem{proposition}[theorem]{Proposition}
\newcommand{\qed}{\hspace*{\fill}\ensuremath{\Box}}
\newenvironment{proof}{\pagebreak[3]\noindent\textbf{Proof.}}{\qed\pagebreak[3]\medskip}

\setlist{itemsep=1pt,parsep=0pt,topsep=2pt}

\algnewcommand{\LineIf}[2]{\State \algorithmicif\, #1 \,\algorithmicthen\, #2 \,\algorithmicend\ \algorithmicif}
\algnewcommand{\LineIfElse}[3]{\State \algorithmicif\, #1 \,\algorithmicthen\, #2 \,\algorithmicelse\, #3 \,\algorithmicend\ \algorithmicif}
\algnewcommand{\LineForAll}[2]{\State \algorithmicforall\, #1 \,\algorithmicdo\, #2 \,\algorithmicend\ \algorithmicfor}

\renewcommand{\geq}{\geqslant}
\renewcommand{\leq}{\leqslant}
\renewcommand{\ge}{\geq}
\renewcommand{\le}{\leq}

\newcommand{\union}{\mathbin{\cup}}
\newcommand{\bigunion}{\mathop{\bigcup}}
\newcommand{\intersect}{\mathbin{\cap}}
\newcommand{\Find}{\ensuremath{\mathrm{Find}}}
\newcommand{\Union}{\ensuremath{\mathrm{Union}}}
\newcommand{\Split}{\ensuremath{\mathrm{Split}}}

\newcommand{\abs}[1] {\ensuremath\left|#1\right|}
\newcommand{\set}[2] {\ensuremath{\left\{#1 \mid #2\right\}}}
\newcommand{\os}[1] {\ensuremath{\left\{#1\right\}}}
\newcommand{\Req} {\mathrel{\ensuremath{\mathcal{R}}}}
\newcommand{\Leq} {\mathrel{\ensuremath{\mathcal{L}}}}

\newcommand{\Synt}{{A^+} / {\equiv_L}}

\newcommand{\autA}{\mathcal{A}}
\newcommand{\partP}{\mathcal{P}}
\newcommand{\queue}{T}
\newcommand{\bigO}[1]{\ensuremath{\mathcal{O}({#1})}}

\newcommand{\ie}{i.e.,~}
\newcommand{\eg}{e.g.~}

\title{Efficient Algorithms for Morphisms over Omega-Regular Languages}
\author{Lukas Fleischer \and Manfred Kuf\-leitner}
\date{FMI, University of Stuttgart\thanks{This work was supported by the DFG grants DI 435/5-2 and \mbox{KU 2716/1-1}.}\\[.1mm]
  \normalsize\texttt{\{fleischer,kufleitner\}@fmi.uni-stuttgart.de}}

\begin{document}

\maketitle

\vspace{-8mm}
\begin{abstract}
  \noindent
  {\sffamily\normalsize\bfseries{Abstract.}} \ 
  Morphisms to finite semigroups can be used for recognizing omega-regular
  languages. The so-called strongly recognizing morphisms can be seen as a
  deterministic computation model which provides minimal objects (known as the
  syntactic morphism) and a trivial complementation procedure. We give a
  quadratic-time algorithm for computing the syntactic morphism from any given
  strongly recognizing morphism, thereby showing that minimization is easy as
  well.
  In addition, we give algorithms for efficiently solving various decision
  problems for weakly recognizing morphisms. Weakly recognizing morphism are
  often smaller than their strongly recognizing counterparts.
  Finally, we describe the language operations needed for converting formulas
  in monadic second-order logic (MSO) into strongly recognizing morphisms, and
  we give some experimental results.
\end{abstract}

\section{Introduction}

Automata over finite words have a huge number of effective closure properties.
Moreover, many problems such as minimization or equivalence of deterministic
automata admit very efficient algorithms~\cite{Hopcroft71,HopcroftKarp71}. The
situation over infinite words is quite similar, but with the major difference
that many operations are less efficient. There are many different automaton
models for accepting languages of infinite words, the so-called
$\omega$-regular languages. Each of these models has its advantages and
disadvantages. For instance, deterministic B\"uchi automata are less powerful
than nondeterministic B\"uchi automata~\cite{Thomas90}. And only very few
automaton models admit efficient minimization algorithms; for example,
minimization of deterministic finite automata can be applied to the lasso
automata in~\cite{CalbrixEtAl1994}. 

The theory of finite semigroups and automata is tightly connected~\cite{rs59}.
Since the semigroup for a language can be exponentially bigger than its
automaton, semigroups have very rarely been considered in the context of
efficient algorithms. There is also an algebraic approach to $\omega$-regular
languages by using morphisms to finite semigroups, see
\eg\cite{PerrinPin04,Thomas90}. Among the many nice properties of this approach
are minimal morphisms\,---\,the so-called syntactic morphisms\,---\,and easy
complementation. As for finite words, the semigroup for an $\omega$-regular
language can be exponentially bigger than its B\"uchi automaton. However, since
many operations for $\omega$-regular languages are less efficient than for
regular languages over finite words, the drawback of this exponential blow-up
in size is less serious. This is even more so when minimizing all intermediate
objects.

A typical algorithm for computing the syntactic morphism of a regular language
over finite words is to minimize the (deterministic) automaton defined by the
Cayley graph of a morphism, and then the syntactic morphism is given by the
transition semigroup of the minimal automaton. This approach does not work for
infinite words and we therefore give a direct algorithm for computing the
syntactic morphism. Our algorithm is an adaptation of Hopcroft's minimization
algorithm~\cite{Hopcroft71} and its running time is quadratic in the size of
the semigroup. We show that this is rather optimal. 

There are two different modes for recognizing omega-regular languages by a
morphism to a finite semigroup: \emph{weak} and \emph{strong} recognition.
Strong recognition is a special case of weak recognition. Easy complementation
and the computation of the syntactic morphism only works for strong
recognition. We show how to test whether a given weak recognition is actually
strong. Another useful tool for morphisms is the computation of the so-called
conjugacy classes.

As an application, we consider the translation of MSO formulas into strongly
recognizing morphisms. To this end, we show that a powerset construction
preserves strong recognition, and that this construction can be used for
computing the image under a length-preserving morphism. Finally, we give the
test results of some translations from MSO to strong recognition. Deciding the
satisfiability of an MSO formula is non-elementary~\cite{Sto74} and therefore,
minimization of intermediate objects is usually very helpful for solving some
special cases. This is confirmed by our test results.

\vspace{-1mm}
\section{Preliminaries}
\label{sec:prelim}

\vspace{-1mm}
\paragraph{Words.}

Let $A$ be a finite \emph{alphabet}. The elements of $A$ are called
\emph{letters}.
A \emph{finite word} is a sequence $a_1 a_2 \cdots a_n$ of letters of $A$ and
an \emph{infinite word} is an infinite sequence $a_1 a_2 \cdots$. The empty
word is denoted by $\varepsilon$. The set of non-empty finite words over $A$ is
$A^+$.
Let $K$ be a set of finite words and let $L$ be a set of infinite words. We set
$KL = \set{u\alpha}{u \in K, \alpha \in L}$, $K^+ = \set{u_1 u_2 \cdots u_n}{n
\ge 1, u_i \in K}$ and $K^* = K^+ \union \os{\varepsilon}$.
Moreover, if $\varepsilon \not\in K$ we define the \emph{infinite iteration}
$K^\omega = \set{u_1 u_2 \cdots}{u_i \in K}$. A natural extension to $K
\subseteq A^*$ is $K^\omega = (K \setminus \os{\varepsilon})^\omega$.

\paragraph{Finite semigroups.}

Let $S$ be a finite semigroup. An element $e$ of $S$ is \emph{idempotent} if
$e^2 = e$. The set of idempotent elements of $S$ is denoted by $E(S) = \set{e
\in S}{e^2 = e}$. For each $s \in S$ the set $\set{s^k}{k \ge 1}$ of all powers
of $s$ is finite and it contains exactly one idempotent element.

A semigroup $S$ is called \emph{$X$-generated} if $X$ is a subset of $S$ and
every element of $S$ can be written as a product of elements of $X$.
The \emph{right Cayley graph} of an $X$-generated semigroup $S$ has $S$ as
vertices and its labeled edges are the triples of the form $(s, a, sa)$ for $s
\in S$ and $a \in X$.
The \emph{left Cayley graph} of $S$ is defined analogously with edges of the
form $(s, a, as)$.
The definitions of Cayley graphs depend on the choice of the set $X$. In the
following, when a surjective morphism $h \colon A^+ \to S$ is given, we choose
$X = h(A)$ as the set of generators.

\emph{Green's relations} are an important tool in the study of finite
semigroups.
We denote by $S^1$ the monoid that is obtained by adding a new neutral element
$1$ to $S$.
For $s, t \in S$ let
\begin{ceqn}
  \begin{align*}
    s & \Req t \text{~if there exist~} q, q' \in S^1 \text{~such that~} sq = t \text{~and~} tq' = s, \\
    s & \Leq t \text{~if there exist~} p, p' \in S^1 \text{~such that~} ps = t \text{~and~} p't = s.
  \end{align*}
\end{ceqn}
These relations are equivalence relations and the equivalence classes of $\Req$
(resp.~$\Leq$) are called \emph{$\Req$-classes} (resp.~\emph{$\Leq$-classes}).
The $\Req$-classes (resp.~$\Leq$-classes) of a semigroup $S$ can be computed in
time linear in $\abs{S}$ by applying Tarjan's algorithm to the right
(resp.~left) Cayley graph of $S$, see~\cite{FroidurePin97}.

An element $(s,e) \in S \times E(S)$ is a \emph{linked pair} if $se = s$.
Two linked pairs $(s, e)$ and
$(t, f)$ are \emph{conjugate}, written as $(s, e) \sim (t, f)$, if there exist
$x, y \in S$ such that $sx = t$, $xy = e$ and $yx = f$. The
conjugacy relation $\sim$ on the set of linked pairs is an equivalence
relation, see \eg\cite{PerrinPin04}. The equivalence classes of $\sim$ are
called \emph{conjugacy classes}.
A set $P$ of linked pairs is \emph{closed under conjugation} if it is a union
of conjugacy classes.

\paragraph{Recognition by morphisms.}

A language $L \subseteq A^\omega$ is \emph{regular} (or
\emph{$\omega$-regular}) if it is recognized by some finite B\"uchi automaton,
see \eg\cite{DiekertGastin08}.
The family of regular languages is closed under Boolean operations, {\ie}set
union, set intersection and complementation.
We now describe algebraic recognition modes for regular languages.
Let $h \colon A^+ \to S$ be a morphism onto a finite semigroup $S$.
For $s \in S$, we set $[s] = h^{-1}(s)$ and for $P \subseteq S \times S$, we
set
\begin{ceqn}
  \begin{align*}
    [P] = \bigunion_{(s, t) \in P} [s][t]^\omega
  \end{align*}
\end{ceqn}
if $h$ is understood from the context.
A language $L \subseteq A^\omega$ is \emph{weakly recognized} by a morphism
$h: A^+ \to S$ if there exists a set of linked pairs $P \subseteq S \times
E(S)$ with $L = [P]$.
If in addition $P$ is closed under conjugation, then $h$ \emph{strongly
recognizes} $L$.
Another well-known characterisation of strong recognition is the following.
\begin{proposition}
  Let $h: A^+ \to S$ be a morphism onto a finite semigroup. Then $h$
  strongly recognizes $L$ if and only if $[s][t]^\omega \intersect L \ne
  \emptyset$ implies $[s][t]^\omega \subseteq L$ for all $s, t \in S$.
  \label{prop:strong}
\end{proposition}
\begin{proof}
  For the direction from left to right, we have $L = [P]$ for some set $P$ that
  is closed under conjugation. Let $\alpha, \beta \in [s][t]^\omega$ for some
  $s, t \in S$ and let $n \ge 1$ such that $t^n \in E(S)$.
  Note that $(st^n, t^n)$ is a linked pair and we also have $\alpha, \beta \in
  [st^n][t^n]^\omega$.
  It suffices to show that $\alpha \in L$ implies $\beta \in L$.
  If $\alpha \in L$, there exist a linked pair $(r, e) \in P$ and a
  factorization $\alpha = u v_1 v_1' v_2 v_2' \cdots$ with $h(u) = st^n$,
  $h(uv_1) = r$, $h(v_i v_i') = t^n$ and $h(v_i' v_{i+1}) = e$ for all $i \ge
  1$.
  Additionally, since $S$ is finite, there exist indices $i, j$ with $1 \le i <
  j$ such that $h(v_i) = h(v_j)$.
  We set $x = h(v_i) = h(v_j)$ and $y = h(v_i' v_{i+1} \cdots v_{j-1}
  v_{j-1}')$. Now, $st^n x = st^{ni}x = h(u v_1 v_1' \cdots v_{i-1} v_{i-1}'
  v_i) = re^{i-1} = r$. By a similar argument, we get $xy = t^n$ and $yx = e$.
  Thus, $(st^n, t^n)$ is contained in $P$ and we have $\beta \in L$.

  For the converse implication, we define $P$ as the union of all linked pairs
  $(s, e)$ with $[s][e]^\omega \subseteq L$. Let $(s, e) \in P$ and let $(t,
  f)$ be a linked pair such that $(s, e)$ and $(t, f)$ are conjugate, \ie$sx =
  t$, $xy = e$ and $yx = f$ for some $x, y \in S$.
  Since $h$ is onto, there exist words $u, v, w \in A^+$ such that $h(u) = s$,
  $h(v) = x$ and $h(w) = y$.
  Now, the infinite word $u(vw)^\omega = uv(wv)^\omega$ is contained in the
  intersection $[s][e]^\omega \intersect [t][f]^\omega$ and by assumption we
  have $[t][f]^\omega \subseteq L$. This shows that $(t, f)$ is in $P$.
\end{proof}

The \emph{syntactic congruence} $\equiv_L$ of a language $L \subseteq A^\omega$
is defined over $A^+$ as $u \equiv_L v$ if the equivalences
\begin{ceqn}
  \begin{align*}
    (xuy)z^\omega \in L & \Leftrightarrow (xvy)z^\omega \in L \text{~and} \\
    z(xuy)^\omega \in L & \Leftrightarrow z(xvy)^\omega \in L
  \end{align*}
\end{ceqn}
hold for all finite words $x, y, z \in A^*$.
Our definition is slightly different but equivalent to the syntactic
congruence introduced by Arnold~\cite{Arnold85}.
The congruence classes of $\equiv_L$ form the so-called \emph{syntactic
semigroup} $\Synt$ and the \emph{syntactic morphism} $h_L \colon A^+ \to \Synt$
is the natural quotient map.
If $L$ is regular, the syntactic semigroup of $L$ is finite and $h_L$ strongly
recognizes $L$~\cite{Arnold85,PerrinPin04}.

\paragraph{Model of computation.}

Morphisms $h \colon A^+ \to S$ are given implicitly through a mapping $f
\colon A \to S$ with $f(a) = h(a)$ for all $a \in A$.
We assume that for finite semigroups $S$, multiplications can be performed in
constant time. Some algorithms only perform multiplications of the form $h(a)
\cdot s$ or $s \cdot h(a)$ where $h$ is a morphism, $s$ is an element of
$S$ and $a$ is a letter.
In that case, semigroups can be represented efficiently by their left and right
Cayley graphs.
For two elements $s, t \in S$ we can check in constant time whether $s = t$ and
it is possible to organize elements of $S$ in a hash map such that operations
on subsets of $S$ can be implemented efficiently.
When a set $P \subseteq S \times S$ is part of the input, we assume that for
each $s, t \in S$ one can check in constant time whether $(s, t) \in P$.

\section{Conversion between B\"uchi automata, weak and strong recognition}

In this section, we describe well-known constructions for the conversion
between the different acceptance modes for regular languages.
For details and proofs, we refer to~\cite{PerrinPin04,Pecuchet86,Thomas90}.

\subsection{From B\"uchi automata to strong recognition}

In the case of finite words, when proving that each regular language is
recognizable by a morphism onto a finite semigroup, one usually considers
the transition semigroup of a finite automaton. However, when applying the same
construction to B\"uchi automata, the resulting morphism only weakly
recognizes the language. In this section, we describe a construction to convert
a B\"uchi automaton $\autA = (Q, A, \delta, I, F)$ into a semigroup $S$ and a
morphism $h \colon A^+ \to S$ that strongly recognizes $L(\autA)$.

For states $p, q \in Q$ and a finite word $u \in A^+$, we write $p
\xrightarrow{u} q$ if there exists a sequence $q_0 a_1 q_1 a_2 q_2 \cdots
q_{n-1} a_n q_n$ with $q_0 = p$, $q_n = q$ and $(q_i, a_{i+1}, q_{i+1}) \in
\delta$ for all $i \in \os{0, \dots, n-1}$. If, additionally, $q_i \in F$ for
some $i \in \os{0, \dots, n}$, we write $p \xrightarrow[F]{u} q$.
We now assign to each word $u \in A^+$ a $Q \times Q$ matrix $h(u)$ defined by
\begin{ceqn}
  \begin{align*}
    (h(u))_{pq} = 
    \begin{cases}
      1 & \text{if~} p \xrightarrow{u} q \text{~but not~} p \xrightarrow[F]{u} q \\
      2 & \text{if~} p \xrightarrow[F]{u} q \\
      0 & \text{otherwise}
    \end{cases}
  \end{align*}
\end{ceqn}
A routine verification shows that this naturally extends
the image of $A^+$ under $h$ to a semigroup $S$. We say that a linked pair $(R,
E)$ where $R = (r_{pq})_{p, q \in Q}$ and $E = (e_{pq})_{p, q \in Q}$ is
\emph{accepting} if there exist states $p, q \in Q$ such that $r_{pq} \ge 1$
and $e_{qq} = 2$. One can now verify that the set $P$ of all accepting linked
pairs is closed under conjugation and that $[P] = L(\autA)$.

\subsection{From weak recognition to B\"uchi automata}

Suppose we are given a morphism $h \colon A^+ \to S$ onto a finite semigroup
$S$ that weakly recognizes a language $L$, \ie$L = [P]$ for some set of linked
pairs $P \subseteq S \times E(S)$. One can use the following construction
from~\cite{Pecuchet86} to obtain a B\"uchi automaton $\autA$ with $L(\autA) =
L$.

The set of states is $Q = S^1 \times E(S)$, the set of initial states is $I =
P$ and the set of final states is $F = \os{1} \times E(S)$. The transition
relation $\delta$ consists of all tuples of the form $((s, e), a, (t, e)) \in Q
\times A \times Q$ where $h(a) t = s$ or $h(a) t = se$.

By combining the constructions from this and the previous subsection, we also
obtain a construction to convert a morphism that weakly recognizes a
language $L$ into a morphism that strongly recognizes $L$. There are also
direct, more efficient constructions, to perform this conversion, see
\eg\cite{PerrinPin04}.
The converse direction is trivial since, by definition, a morphism $h
\colon A^+ \to S$ that strongly recognizes a language $L$ also weakly
recognizes $L$.

\section{Computing conjugacy classes}

When designing an algorithm that takes a set of linked pairs $P \subseteq S
\times E(S)$ as input, it is often convenient to assume that $P$ is closed
under conjugation.
However, this is not always the case in practice: The input set $P$ might be a
proper subset of its closure under conjugation $Q$ such that $[P] = [Q]$.
In this section, we describe an algorithm to compute the conjugacy classes
efficiently. It justifies the assumption that $P$ is always closed under
conjugation in the following sections, particularly in
Section~\ref{sec:syntactic}.

As a warm-up, we first describe how to compute the set $F$ of linked pairs.
The linked pairs are exactly the pairs of the form $(se, e)$ with $s \in S$ and
$e \in E(S)$. Thus, we first check for each element $e \in S$ whether $e^2 =
e$.
If the outcome of the check is positive, we perform a depth-first search in the
left Cayley graph of $S$, starting at element $e$. For each element $s$ that is
visited, $(s, e)$ is a linked pair.
The total running time of this routine is $\bigO{\abs{S} + \abs{A} \cdot
\abs{F}}$.

An equivalence relation $\equiv$ on the set of linked pairs is called
\emph{left-stable} if for all $p \in S$ and for linked pairs $(s, e)$, $(t, f)$
with $(s, e) \equiv (t, f)$, we have $(ps, e) \equiv (pt, f)$.
We define an equivalence relation $\approx$ on the set of linked pairs by $(s,
e) \approx (t, f)$ if and only if $e \Leq s \Req t \Leq f$ or $(s, e) = (t,
f)$. Its relationship to conjugacy is captured in the following Lemma:
\begin{lemma}
  The conjugacy relation $\sim$ is the finest left-stable equivalence relation
  coarser than $\approx$.
  \label{lem:simapprox}
\end{lemma}
\begin{proof}
  It follows directly from the definitions of linked pairs and conjugacy that
  $\sim$ is left-stable.
  Let $(s, e)$ and $(t, f)$ be linked pairs with $(s, e) \approx (t, f)$ and
  $(s, e) \ne (t, f)$. Since $s \Req t$, there exist $q, q' \in S^1$ such that
  $sq = t$ and $tq' = s$. We set $x = eq$ and $y = fq'$. Now, $sx = seq = sq =
  t$. Moreover, since $s \Leq e$, there exists $p \in S^1$ with $ps = e$. Thus,
  we have $xy = eqy = psqy = pty = ptfq' = ptq' = ps = e$. A similar argument
  can be used to show that $yx = f$. Hence, $(s, e)$ and $(t, f)$ are
  conjugate, and $\sim$ is indeed coarser than $\approx$.

  In order to show that $\sim$ is the finest relation with these properties, we
  consider an arbitrary left-stable equivalence relation $\simeq$ on the set of
  linked pairs which is coarser than $\approx$. We show that $(s, e) \sim (t,
  f)$ implies $(s, e) \simeq (t, f)$. Let $x, y \in S$ such that $sx = t$, $xy
  = e$ and $yx = f$. Then we have $ex = xyx = xf$ and $xfy = xyxy = e^2 = e$,
  which shows that $e \Req xf$. Furthermore we have $xf \Leq f$, since $yxf =
  f^2 = f$. By the definition of $\approx$, this means that $(e, e) \approx
  (xf, f)$ and since $\approx$ refines $\simeq$, it follows that $(e, e) \simeq
  (xf, f)$. Left-stability yields $(s, e) = (se, e) \simeq (sxf, f) = (t, f)$.
\end{proof}

Since $\Req$-classes and $\Leq$-classes can be computed in time linear in the
size of the semigroup, this allows us to efficiently compute the conjugacy
classes as shown in Algorithm~\ref{alg:conjugacy}. We use a so-called
\emph{disjoint-set data structure} that provides two operations on a
partition.
$\Find(s, e)$ returns a unique element from the class that contains $(s, e)$,
{\ie}if $(s, e)$ and $(t, f)$ are in the same class, we have $\Find(s, e) =
\Find(t, f)$.  $\Union((s, e), (t, f))$ merges the classes of $(s, e)$ and $(t,
f)$.
To simplify the notation we also introduce an operation $\Union^+(R)$ for
subsets $R$ of $S \times S$ that merges all classes with elements in $R$.
$\Union^+(R)$ can be implemented using $\abs{R}-1$ atomic $\Union$ operations.
The partition is initialized with singleton sets $\os{(s, e)}$ for all linked
pairs $(s, e)$.
The second data structure used in the algorithm is a set $\queue \subseteq
2^F$.
\begin{algorithm}[h!]
  \caption{Computing conjugacy classes}
  \begin{algorithmic}
    \State initialize $\queue$ with the non-trivial equivalence classes of $\approx$
    \LineForAll{$R \in \queue$}{$\Union^+(R)$}
    \While{$\queue \ne \emptyset$}
      \State remove some set $R$ from $\queue$
      \ForAll{$a \in A$}
        \State $R' \gets \emptyset$
        \LineForAll{$(s, e) \in R$}{$R' \gets R' \union \os{\Find(h(a)s, e)}$}
        \If{$\abs{R'} > 1$}
          \State $\Union^+(R')$
          \State $\queue \gets \queue \union \os{R'}$
        \EndIf
      \EndFor
    \EndWhile
  \end{algorithmic}
  \label{alg:conjugacy}
\end{algorithm}

To prove the correctness and running time of the algorithm, one can combine
Lemma~\ref{lem:simapprox} with arguments similar to those given in the
correctness and running time proofs of the \emph{Hopcroft-Karp equivalence
test}~\cite{HopcroftKarp71}. We first show that the relation induced by the
final partition is left-stable:
\begin{lemma}
  Let $(s, e)$ and $(t, f)$ be linked pairs of the same class upon termination,
  then, for each $a \in A$, the pairs $(h(a)s, e)$ and $(h(a)t, f)$ are in the
  same class as well.
  \label{lem:left-stable}
\end{lemma}
\begin{proof}
  We write $\Find_i(s, e) = \Find_i(t, f)$ if $(s, e)$ and $(t, f)$ belong to
  the same class after the $i$-th iteration of the \emph{while}-loop.
  The index $\infty$ is used to describe the situation upon termination.

  Let $i$ be minimal such that for some pairs $(s, e), (t, f)$ and a letter $a
  \in A$, we have $\Find_i(s, e) = \Find_i(t, f)$ and $\Find_\infty(h(a)s, e)
  \ne \Find_\infty(h(a)t, f)$.
  Note that $i > 0$ because otherwise, a set containing both $(s, e)$ and $(t,
  f)$ would be added to $\queue$ during initialization.
  Hence, there exists a pair $(s', e')$ with $\Find_{i-1}(s', e') =
  \Find_{i-1}(s, e)$ and a pair $(t', f')$ with $\Find_{i-1}(t', f') =
  \Find_{i-1}(t, f)$ such that $\Union^+(R)$ is executed for some set $R
  \supseteq \os{(s', e'), (t', f')}$.
  By choice of $i$, we have $\Find_\infty(h(a)s, e) = \Find_\infty(h(a)s', e')$
  and $\Find_\infty(h(a)t, f) = \Find_\infty(h(a)t', f')$.
  Since we add the set $R$ to $\queue$ in iteration $i$, the equality
  $\Find_\infty(h(a)s', e') = \Find_\infty(h(a)t', f')$ holds as well, and thus
  $\Find_\infty(h(a)s, e) = \Find_\infty(h(a)t, f)$, a contradiction.
\end{proof}

There is of course a dual statement for the pairs $(s \cdot h(a), e)$ and $(t
\cdot h(a), f)$.
\begin{theorem}
  Let $F$ be the set of linked pairs of $S$. When Algorithm~\ref{alg:conjugacy}
  terminates, the classes of the partition correspond to the conjugacy classes
  of $F$.
  Furthermore, the algorithm executes at most
  \begin{itemize}[label=$\blacktriangleright$]
    \item $\abs{F} - 1$ $\Union$ operations and
    \item $2\abs{A}(\abs{F} - 1)$ $\Find$ operations.
  \end{itemize}
  \label{thm:conjugacy}
\end{theorem}
\begin{proof}
  By Lemma~\ref{lem:left-stable}, the relation induced by the final partition
  is left-stable and throughout the main algorithm, two classes are only merged
  when required to establish this property. Thus, the relation is the finest
  left-stable equivalence relation coarser than $\approx$ and, by
  Lemma~\ref{lem:simapprox}, equivalent to the conjugacy relation.

  The number of $\Union$ operations is bounded by $\abs{F} - 1$ since each
  operation reduces the number of classes in the partitions by $1$.
  Let $R_1, \dots, R_k$ be the sets that are added to $\queue$ during the
  execution of the algorithm. Whenever one of the sets $R_i$ is inserted into
  $\queue$, $\abs{R_i} - 1$ $\Union$ operations are executed. Thus, we have
  \begin{ceqn}
    \begin{equation*}
      \sum_{i=1}^k \bigl(\abs{R_i} - 1\bigr) \le \abs{F} - 1.
    \end{equation*}
  \end{ceqn}
  When $R_i$ is removed from $\queue$, exactly $\abs{A} \cdot \abs{R_i}$
  $\Find$ operations are executed in the same iteration of the
  \emph{while}-loop. The total number of $\Find$ operations is therefore
  bounded by
  \begin{ceqn}
    \begin{equation*}
      \sum_{i=1}^k \abs{A} \cdot \abs{R_i} \le
      \sum_{i=1}^k \abs{A} \cdot (2\abs{R_i} - 2) \le
      2 \abs{A} \cdot (\abs{F} - 1)
    \end{equation*}
  \end{ceqn}
  where the first inequality follows from the fact that each of the sets $R_i$
  contains at least two elements.
\end{proof}

A sequence of $n$ $\Union$- and $m$ $\Find$-operations can be performed in
$\bigO{n + m \cdot \alpha(n)}$ time where $\alpha(n)$ denotes the extremely
slow-growing \emph{inverse Ackermann function}~\cite{Tarjan75}.
Thus, when considering a fixed-size alphabet, the total running time of our
algorithm is ``almost linear'' in the number of linked pairs.

\section{Testing for strong recognition}
\label{sec:strong}

Common decision problems, such as the \emph{universality problem} or the
\emph{inclusion problem}, are easy in the case of strong recognition. In the
context of weak recognition, the algorithm presented in this section is a
powerful tool to answer a broad range of similar problems. Given a morphism
$h \colon A^+ \to S$ onto a finite semigroup $S$ and two sets of linked pairs $P, Q \subseteq S
\times E(S)$, it can be used to check whether $[P] \subseteq [Q]$. In
particular, it allows for testing whether the morphism strongly recognizes
a language $L = [P]$ by first computing the closure $Q$ of $P$ under
conjugation and then using the algorithm to test whether $[Q] \subseteq [P]$.

Before we present the algorithm, we remark that inclusion is not only a
property of the semigroup $S$ and the sets $P$ and $Q$ but it also depends on
the set of generators $h(A)$. In order to see this, we consider the finite
semigroup $S = \set{(i, j)}{1 \le i, j \le 2}$ with the multiplication given by
$(i, j) \cdot (k, \ell) = (i, \ell)$ for all $i, j, k, \ell \in \os{1, 2}$. Let
$A = \os{a, b}$ and let $h \colon A^+ \to S$ be the surjective morphism defined
by $h(a) = (1, 2)$ and $h(b) = (2, 1)$. We consider the two sets of linked
pairs $P = \os{((1, 1), (1, 1))}$ and $Q = \os{((1, 2), (2, 2))}$. It is easy
to check that $[P] = [Q] = (a^+b^+)^\omega$. However, if we add a new letter
$c$ to $A$ and extend $h$ by setting $h(c) = (1, 1)$, the infinite word
$c^\omega$ is contained in $[P]$ but not in $[Q]$, which implies $[P]
\not\subseteq [Q]$. The morphism (or another description of the set of
generators) thus needs to be part of the input of any algorithm performing the
inclusion test described above.

Let us now describe the algorithm.
It maintains two sets $R, T \subseteq S \times S^1 \times S^1$. The former
keeps record of the elements that are added to $T$ during the course of the
algorithm.
To simplify the presentation, we define $x \cdot a^{-1}$ to be the set of all
elements $p \in S^1$ which satisfy the equation $p \cdot h(a) = x$.

\begin{algorithm}[h!]
  \caption{Testing for strong recognition}
  \begin{algorithmic}
    \State initialize $R$ and $\queue$ with the set $\set{(s, e, 1)}{(s, e) \in P}$
    \While{$\queue \ne \emptyset$}
      \State remove some element $(s, x, y)$ from $\queue$
      \LineIf{$x = 1$}{\Return ``$[P] \not\subseteq [Q]$''}
      \If{$(sx, yxyx) \not\in Q$}
        \ForAll{$a \in A$, $p \in x \cdot a^{-1}$}
          \LineIf{$(s, p, h(a)y) \not\in R$}{add $(s, p, h(a)y)$ to $R$ and to $\queue$}
        \EndFor
      \EndIf
    \EndWhile
    \State \Return ``$[P] \subseteq [Q]$''
  \end{algorithmic}
  \label{alg:teststrong}
\end{algorithm}

\noindent
The following technical Lemma is crucial for the correctness proof of the
algorithm:
\begin{lemma}
  Let $u, v \in A^+$ and let $(s, e)$ and $(h(u), h(v))$ be linked pairs.
  Then $uv^\omega$ is contained in $[s][e]^\omega$ if and
  only if there exists a factorization $v = v_1 v_2$ such that $v_1 \ne
  \varepsilon$, $h(u v_1) = s$ and $h(v_2 v v_1) = e$.
  \label{lem:lpfac}
\end{lemma}
\begin{proof}
  Let $v = a_1 a_2 \cdots a_n$ with $n \ge 1$ and $a_i \in A$. If $uv^\omega$
  is contained in $[s][e]^\omega$, there exists a factorization $uv^\omega = u'
  v_1' v_2' \cdots$ such that $h(u') = s$ and $h(v_i') = e$ for all $i \ge 1$.
  Since $u$ and $v$ are finite words, there exist indices $j > i \ge 1$, powers
  $k, \ell \ge 1$ and a position $m \in \os{1, \dots, n}$ such that $u' v_1'
  v_2' \cdots v_{i-1}' = u v^k a_1 a_2 \cdots a_m$ and $v_i' v_{i+1}' \cdots
  v_j' = a_{m+1} a_{m+2} \cdots a_n v^\ell a_1 a_2 \cdots a_m$.
  We set $v_1 = a_1 a_2 \cdots a_m$ and $v_2 = a_{m+1} a_{m+2} \cdots a_n$.
  Then $v_1 v_2 = v$,
  \begin{ceqn}
    \begin{alignat*}{4}
      h(u v_1) &= h(u v^k a_1 a_2 \cdots a_m) & &= h(u' v_1' v_2' \cdots v_{i-1}') & &= se^{i-1} & &= s \text{~and} \\
      h(v_2 v v_1) &= h(a_{m+1} a_{m+2} \cdots a_n v^\ell a_1 a_2 \cdots a_m) & &= h(v_i' v_{i+1}' \cdots v_j') & &= e^{j-i+1} & &= e.
    \end{alignat*}
  \end{ceqn}
  To prove the converse direction, consider the factorization $uv^\omega = uv_1
  (v_2 v v_1)^\omega$.
\end{proof}

To simplify the proofs of the following two Lemmas, we extend $h$ to a monoid
morphism $h^1 \colon A^* \to S^1$ by setting $h^1(u) = h(u)$ for all $u \in
A^+$ and $h^1(\varepsilon) = 1$.
\begin{lemma}
  If the difference $[P] \setminus [Q]$ is non-empty, the algorithm returns
  ``$[P] \not\subseteq [Q]$''.
  \label{lem:strong1}
\end{lemma}
\begin{proof}
  By the closure properties of regular languages, we know that there exists a
  word $\alpha = u(a_1 a_2 \cdots a_n)^\omega \in [P] \setminus [Q]$.  Let $s =
  h(u)$ and $e = h(a_1 a_2 \cdots a_n)$. Lemma~\ref{lem:lpfac} shows that we
  can assume without loss of generality that $(s, e)$ is contained in $P$.
  We now prove by induction on the parameter $k$ that upon termination, we have
  $(s, h^1(a_1 a_2 \cdots a_k), h^1(a_{k+1} a_{k+2} \cdots a_n)) \in R$ for all
  $k \in \os{0, \dots, n}$.
  In particular, by considering the case $k = 0$, we see that the element $(s,
  1, e)$ is added to $R$. Since every element added to $R$ is also added to
  $Q$, the algorithm returns ``$[P] \not\subseteq [Q]$''.
  
  The base case $k = n$ is covered by the initialization of the set $R$.
  Let now $k < n$, $x = h^1(a_1 a_2 \cdots a_{k+1})$ and $y = h^1(a_{k+2}
  a_{k+3} \cdots a_n)$.
  By the induction hypothesis, we know that the tuple $(s, x, y)$ is added to
  $\queue$ during the course of the algorithm. Consider the iteration when
  this tuple is removed from $\queue$.
  Because of $\alpha \not\in [Q]$, we know that $(sx, yxyx) \not\in Q$. Thus
  the inner loop guarantees that $(s, h^1(a_1 a_2 \cdots a_k), h^1(a_{k+1}
  a_{k+2} \cdots a_n))$ is added to $R$.
\end{proof}

\begin{lemma}
  If the algorithm returns ``$[P] \not\subseteq [Q]$'', the difference $[P]
  \setminus [Q]$ is non-empty.
  \label{lem:strong2}
\end{lemma}
\begin{proof}
  We construct a word in the difference $[P] \setminus [Q]$.
  For every triple $(s, e, 1)$ that is added to $R$ during the initialization,
  we define $w[s, e, 1] = \varepsilon$. If a triple $(s, p, h(a)y)$ is added to
  $R$ later, we set $w[s, p, h(a)y] = a \cdot w[s, p \cdot h(a), y]$.
  For every $(s, x, y) \not\in R$, the word $w[s, x, y]$ is undefined.
  For the other words, well-definedness follows from the fact that each triple
  $(s, x, y)$ is added to $R$ at most once.
  Furthermore, if $w[s, x, y]$ is defined, its image under $h^1$ is $y$ and we
  have $(s, xy) \in P$. Both properties are easy to prove by induction.

  Let $(s, 1, y)$ be the triple that was removed from $\queue$ immediately
  before the termination of the algorithm.
  Consider an arbitrary word $u \in [s]$ and set $v = w[s, 1, y]$.
  We have $(s, y) \in P$ and thus $uv^\omega \in [P]$.
  For every factorization $v = v_1 a v_2$ where $v_1, v_2 \in A^*$ and $a \in
  A$, the word $w[s, h^1(v_1), h^1(a v_2)]$ is defined as $a v_2$ and thus, the
  tuple $(h(u v_1 a), h(v_2 v v_1 a))$ is not contained in $Q$. In view of
  Lemma~\ref{lem:lpfac}, this shows that $uv^\omega \not\in [Q]$.
\end{proof}

We are now able to state the main result of this section:
\begin{theorem}
  Given a morphism $h \colon A^+ \to S$ onto a finite semigroup $S$ and two
  sets of linked pairs $P, Q \subseteq S \times E(S)$, one can check in $\bigO{\abs{A} \cdot
  \abs{S}^3}$ time whether $[P] \subseteq [Q]$.
  \label{thm:strong}
\end{theorem}
\begin{proof}
  The correctness of Algorithm~\ref{alg:teststrong} follows from the previous
  two Lemmas.
  Since $R$ contains at most $(\abs{S} + 1)^3$ elements when the algorithm
  terminates, the outer loop is executed at most $(\abs{S} + 1)^3$ times.
  Moreover, for all $a \in A$ and $s, t \in S$ with $s \ne t$, the sets $s
  \cdot a^{-1}$ and $t \cdot a^{-1}$ are disjoint. Thus, each element $p \in
  S^1$ is considered at most $\abs{A} \cdot (\abs{S} + 1)^2$ times in the inner
  loop.
  If $R$ is implemented as a bit field and $\queue$ is implemented as a linked
  list, all operations take constant time.
  This shows that the total running time is in $\bigO{\abs{A} \cdot
  \abs{S}^3}$.
\end{proof}

\section{Computation of the syntactic morphism}
\label{sec:syntactic}

In this section, we present an algorithm to compute the syntactic semigroup for
a given language. The syntactic homomorphism is obtained as a byproduct. One
can show that the syntactic semigroup is the smallest semigroup strongly
recognizing a language~\cite{Arnold85,PerrinPin04}, so this operation is
similar to the minimization of finite automata. The most important difference
is that our algorithm requires only quadratic time, whereas minimization is
\textsf{PSPACE}-hard in the case of B\"uchi
automata~\cite{MeyerStockmeyer72,SistlaVardiWolper87}.

Let $S$ be a finite semigroup, let $h \colon A^+ \to S$ be a surjective
morphism and let $P$ be a set of linked pairs that is closed under conjugation.
To make the following notation more readable, we define $Q$ as the maximal
subset of $S \times S$ such that $[P] = [Q]$.
\begin{lemma}
  Let $u, v \in A^+$. Then $uv^\omega \in [P]$ if and only if $(h(u), h(v)) \in
  Q$.
  \label{lem:maximal}
\end{lemma}
\begin{proof}
  Suppose that $uv^\omega \in [P]$. By Proposition~\ref{prop:strong}, we have
  $[h(u)][h(v)]^\omega \subseteq [P] = [Q]$. Since $Q$ is maximal, the pair
  $(h(u), h(v))$ is contained in $Q$.
  The converse implication is trivial.
\end{proof}

We now define a equivalence relation $\cong$ on $S$ by $s \cong t$ if for
all $z \in S$, we have
\begin{ceqn}
  \begin{align*}
    (z, s) \in Q & \Leftrightarrow (z, t) \in Q \text{~and~} \\
    (s, z) \in Q & \Leftrightarrow (t, z) \in Q.
  \end{align*}
\end{ceqn}
Moreover, let $\equiv$ be the coarsest congruence on $S$ that refines
$\cong$, \ie$s \equiv t$ if $xsy \cong xty$ for all $x, y \in S^1$.
We denote by $[s]_\equiv$ the equivalence class $\set{t \in S}{t \equiv s}$
of an element $s \in S$.
The relation $\equiv$ is closely related to the syntactic congruence, as
confirmed by the following result:
\begin{proposition}
  The quotient semigroup ${S} / {\equiv}$ is isomorphic to $\Synt$.
  \label{prop:syntactic}
\end{proposition}
\begin{proof}
  We first define a morphism $g \colon A^+ \to {S} / {\equiv}$ by setting
  $g(u) = [h(u)]_\equiv$ for all $u \in A^+$. Let now $u, v \in A^+$. By
  Lemma~\ref{lem:maximal}, we have $h(u) \equiv h(v)$ if and only if $h_L(u) =
  h_L(v)$. Thus, $g \circ h_L^{-1}$ is a semigroup isomorphism.
\end{proof}

\noindent
The computation of the syntactic semigroup requires two steps:
\begin{enumerate}
  \item Compute the partition induced by the equivalence relation $\cong$.
  \item Refine the partition until the underlying equivalence relation becomes
    a congruence.
\end{enumerate}
The first step can be performed in time quadratic in the size of the semigroup.
For the second step, we can adapt Hopcroft's minimization algorithm for finite
automata~\cite{Hopcroft71}.
For $C \subseteq S$ and $a \in A$, we define
\begin{ceqn}
  \begin{align*}
    C \cdot a^{-1} = \set{s \in S}{s \cdot h(a) \in C} \quad\text{and}\quad
    a^{-1} \cdot C = \set{s \in S}{h(a) \cdot s \in C}.
  \end{align*}
\end{ceqn}
The full algorithm is shown in Algorithm~\ref{alg:minimization}. It relies on
the $\Split$ routine that is usually implemented as part of a \emph{partition
refinement data structure}, see \eg\cite{Hopcroft71} for details. Its semantics
is shown in Algorithm~\ref{alg:split}.
In addition to modifying the partition, that routine also updates a set $\queue
\subseteq 2^S$ that is used in the main algorithm.
\begin{algorithm}[h!]
  \caption{Computing the syntactic semigroup}
  \begin{algorithmic}
    \State initialize a partition with a single class $S$
    \ForAll{$s \in S$}
      \State $\Split(\set{t \in S}{(s, t) \in Q})$
      \State $\Split(\set{t \in S}{(t, s) \in Q})$
    \EndFor
    \State initialize $\queue$ with the non-trivial classes of the partition
    \While{$\queue \ne \emptyset$}
      \State remove some set C from $\queue$
      \ForAll{$a \in A$}
        \State $\Split(C \cdot a^{-1})$ \Comment Refine the partition and update $\queue$
        \State $\Split(a^{-1} \cdot C)$ \Comment Refine the partition and update $\queue$
      \EndFor
    \EndWhile
  \end{algorithmic}
  \label{alg:minimization}
\end{algorithm}
\begin{algorithm}[h!]
\caption{The $\Split$ operation to refine a partition $\partP$}
\begin{algorithmic}
\Procedure{\Split}{$X$}
  \ForAll{$C \in \partP$}
    \State $C_1 \gets C \intersect X$, $C_2 \gets C \setminus X$
    \If{$C_1 \ne \emptyset$ and $C_2 \ne \emptyset$}
      \State $\partP \gets (\partP \setminus \os{C}) \union \os{C_1, C_2}$
      \If{$C \in \queue$}
        \State $\queue \gets (\queue \setminus \os{C}) \union \os{C_1, C_2}$
      \Else
        \LineIfElse{$\abs{C_1} \le \abs{C_2}$}{$\queue \gets \queue \union \os{C_1}$}{$\queue \gets \queue \union \os{C_2}$}
      \EndIf
    \EndIf
  \EndFor
\EndProcedure
\end{algorithmic}
\label{alg:split}
\end{algorithm}

The next Lemma shows that upon termination, the equivalence relation induced by
the partition is indeed a congruence:
\begin{lemma}
  If, upon termination, the elements $s$ and $t$ belong to the same class of
  the partition, then, for each $a \in A$, the elements $h(a)s$ and $h(a)t$ are
  in the same class as well.
  \label{lem:congruence}
\end{lemma}
\begin{proof}
  Suppose that $h(a) \cdot s$ and $h(a) \cdot t$ belong to different classes.
  These elements are split either during the initialization or in the main
  loop. In either case, a set $C$ that contains either $h(a) \cdot s$ or $h(a)
  \cdot t$ is added to $\queue$. When this set is removed from $\queue$, the
  operation $\Split(a^{-1} \cdot C)$ asserts that $s$ and $t$ lie in different
  classes as well.
\end{proof}

There is of course a dual statement for the elements $s \cdot h(a)$ and $t
\cdot h(a)$.
\begin{theorem}
  \label{thm:syntactic}
  The syntactic morphism can be computed in $\bigO{\abs{S}^2 + \abs{A} \cdot
  \abs{S}\log{\abs{S}}}$ time.
\end{theorem}
\begin{proof}
  Let us first argue that Algorithm~\ref{alg:minimization} is correct.
  The partition is initialized with the equivalence classes of $\cong$. A class
  is only split when it is necessary to restore the left-stability or
  right-stability. Upon termination, the relation induced by the partition is a
  congruence, as stated in Lemma~\ref{lem:congruence}. Thus, it is the coarsest
  congruence that refines $\cong$ and hence equivalent to $\equiv$.

  For the analysis of the running time, we assume that the operation
  $\Split(X)$ can be implemented in time linear in $\abs{X}$.
  Then the initialization clearly takes $\bigO{\abs{S}^2}$ time.
  We denote by $C_1, \dots, C_k$ the sets that are added to $\queue$ during
  the course of the algorithm. Let $s \in S$ and let $n_s = \set{i}{1 \le i \le
  k, s \in C_i}$ be the number of sets $C_i$ containing $s$.
  At any point in time, there is at most one set in $\queue$ that contains
  $s$. If such a set $C$ is removed from $\queue$ and another set $C'$ with $s
  \in C'$ is added to $\queue$ at a later point in time, we have that
  $\abs{C'} \le \abs{C} / 2$. Thus, the inequality $n_s \le \log{\abs{S}}$
  holds for all $s \in S$ and we have
  \begin{ceqn}
    \begin{equation*}
      \sum_{i=1}^k \sum_{a \in A} \bigl(\abs{C_i \cdot a^{-1}} + \abs{a^{-1} \cdot C_i}\bigr) =
      \sum_{s \in S, a \in A} \bigl(n_{s \cdot h(a)} + n_{h(a) \cdot s}\bigr) \le
      2 \abs{A} \cdot \abs{S} \log{\abs{S}}.
    \end{equation*}
  \end{ceqn}
  Consequently, the total running time of the \emph{while}-loop is in
  $\bigO{\abs{A} \cdot \abs{S} \log{\abs{S}}}$, assuming that $\queue$ is
  implemented efficiently, \eg{}as a linked list.
\end{proof}

If the alphabet $A$ is fixed and the semigroup $S$ becomes large, the running
time is dominated by the initialization. However, one can show that the
algorithm we presented is quite optimal.
Before we start with the proof of the optimality result, we need the following
technical Lemma that asserts the existence of a semigroup with certain
properties:
\begin{lemma}
  For every $n \ge 4$ there exist a semigroup $T$ with $n^2 \cdot 2^n + n$
  elements and a set $D \subseteq E(T)$ such that the following properties
  hold:
  \begin{enumerate}
    \item $T$ has rank $2$, \ie$T$ is $X$-generated for some $X \subseteq T$
      with $\abs{X} = 2$.
    \item $\abs{D} = 2^{n-1}$.
    \item For all $e, f \in D$ and $x, y \in T$, we have $e = f$ or $xy \ne e$ or $yx \ne f$. \label{enum:bbb}
  \end{enumerate}
  \label{lem:existence-semigroup}
\end{lemma}
\begin{proof}
  Let $n \ge 4$ and let $N = \os{0, \dots, n-1}$. Let $T$ be the set $N \times
  2^N \times N \union N$. We denote by $+$ be the addition modulo $n$ which can
  be extended to $T$ as follows:
  \begin{ceqn}
    \begin{align*}
      (i, X, j) + (k, Y, \ell) & = (i, X \union \os{j + k} \union Y, \ell) \\
      (i, X, j) + k & = (i, X, j + k) \\
      i + (j, X, k) & = (i + j, X, k)
    \end{align*}
  \end{ceqn}
  for all $i, j, k, \ell \in N$ and $X, Y \subseteq N$. It is easy to check
  that this operation is associative and thus, $(T, +)$ forms a semigroup.
  The number of elements of $T$ is $n^2 \cdot 2^n + n$.
  One can also easily verify that $T$ is $\os{1, (0, \emptyset, 0)}$-generated.

  Now, consider the set $D$ of all elements of the form $(0, X, 0)$ for
  $0 \in X \subseteq N$. We have $(0, X, 0) \cdot (0, X, 0) = (0, X \union
  \os{0} \union X, 0) = (0, X, 0)$ and thus, $D \subseteq E(T)$.
  The number of elements in $D$ is $2^{n-1}$.
  To show property~\ref{enum:bbb}, we assume that there exist $E, F \subseteq
  N$ and $(i, X, j), (k, Y, \ell) \in T$ such that $(i, X, j) \cdot (k, Y,
  \ell) = (0, E, 0)$ and $(k, Y, \ell) \cdot (i, X, j) = (0, F, 0)$.
  By the definition of the operation $+$ on $T$, this implies $i = j = k = \ell
  = 0$.  Moreover, we have $E = X \union \os{0} \union Y = Y \union \os{0}
  \union X = F$. The other cases ($x \in N$ or $y \in N$) are similar.
\end{proof}

We now use the previous Lemma to construct another semigroup with four
generators and a large number of conjugacy classes.
\begin{lemma}
  Let $A = \os{a, b, \overline a, \overline b}$, let $c \in \mathbb{N}$ and let
  $\lambda \in \mathbb{R}$ be a strictly positive number. Then there exist a
  semigroup $S$ and a surjective morphism $g \colon A^+ \to S$, such that $S$
  has more than $c \cdot \abs{S}^{2-\lambda}$ conjugacy classes.
  \label{lem:existence-semigroup2}
\end{lemma}
\begin{proof}
  We first define $B = \os{a, b}$, $\overline B = \os{\overline a, \overline
  b}$ and choose $n \ge 4$ such that $32cn^2 < 2^{\lambda n}$.
  Let $T$ be a finite semigroup and let $D$ be a subset of $E(T)$ with the
  properties described in Lemma~\ref{lem:existence-semigroup}. Let $h \colon
  B^+ \to T$ be a surjective homomorphism.
  We denote by $\overline T$ a disjoint copy of $T$ and by $\overline h$ the
  morphism $\overline h \colon \overline B \to \overline T$ induced by $h$. Now
  we define $S = ({\overline T}^1 \times 1) \union (\overline 1 \times T^1)
  \setminus \os{(\overline 1, 1)}$ with the multiplication
  \begin{ceqn}
    \begin{align*}
      (\overline s, s) \cdot (\overline t, t) =
      \begin{cases}
        (\overline 1, s \cdot t) & \text{if } \overline s = \overline t = \overline 1 \\
        (\overline s \cdot \overline t, 1) & \text{otherwise}
      \end{cases}
    \end{align*}
  \end{ceqn}
  where $1$ denotes the identity in $T^1 \setminus T$ and $\overline 1$ denotes
  the identity in ${\overline T}^1 \setminus \overline T$.
  By construction, the semigroup $S$ has $2 n^2 2^n + 2n + 1 < 4 n^2 2^n$
  elements.
  The morphism $g \colon A^+ \to S$ defined by $g(\overline c) = (\overline
  h(\overline c), 1)$ for $\overline c \in \overline B$ and $g(c) = (\overline
  1, h(c))$ for $c \in B$ is surjective.

  Consider the set $F = (\overline T \times 1) \times (\overline 1 \times D)$.
  We will show that $F$ contains more than $c \cdot \abs{S}^{2-\lambda}$
  elements, that each element of $F$ is a linked pair of $S$ and that no two
  different elements of $F$ are conjugate, thereby proving the claim.

  We start with the cardinality of $F$. We have $\abs{F} > n^2 2^{2n-1} = n^2
  2^{2n-\lambda n+\lambda n-1} > 16cn^4 (2^n)^{2-\lambda} > c (4n^2
  2^n)^{2-\lambda} > c \abs{S}^{2-\lambda}$, where the second inequality follows
  by the choice of $n$.
  Showing that $F$ only consists of linked pairs is easy and is left as an
  exercise to the reader.
  Now consider two pairs $((\overline s, 1), (\overline 1, e))$ and
  $((\overline t, 1), (\overline 1, f))$ from $F$. Suppose these pairs are
  conjugate, \ie there exist $(\overline x, x), (\overline y, y) \in S$ such
  that
  $(\overline s, 1) \cdot (\overline x, x) = (\overline t, 1)$,
  $(\overline x, x) \cdot (\overline y, y) = (\overline 1, e)$ and
  $(\overline y, y) \cdot (\overline x, x) = (\overline 1, f)$.
  From the second equation, we see that $\overline x = \overline y = \overline
  1$. Therefore, $\overline s = \overline t$.  Additionally, we have $x y = e$,
  as well as $y x = f$. Property~\ref{enum:bbb} in
  Lemma~\ref{lem:existence-semigroup} yields $e = f$.
\end{proof}

The optimality result now follows by using the previous construction as an
input to the minimization algorithm.
\begin{proposition}
  The syntactic morphism cannot be computed in time
  $\bigO{\abs{S}^{2-\lambda}}$ for any strictly positive, fixed value $\lambda
  \in \mathbb{R}$.
\label{prop:optimality}
\end{proposition}
\begin{proof}
  Assume there exists an algorithm and a constant $c \ge 1$ such that every
  input of size $n = \abs{S}$ can be minimized in time $T(n) \le c \cdot
  n^{2-\lambda}$. Consider the execution of the algorithm on the semigroup $S$
  described in Lemma~\ref{lem:existence-semigroup2} and on $P = F$. 
  We denote by $(s_1, e_1), (s_2, e_2), \dots, (s_\ell, e_\ell)$ the sequence
  of linked pairs for which the algorithm checks whether $(s_i, e_i) \in P$.
  We have $\ell \le T(n) \le c \cdot \abs{S}^{2-\lambda}$ and thus, there is a
  conjugacy class $C$ such that $(s_i, e_i) \not\in C$ for all $i \in \os{1,
  \dots, \ell}$.
  Since the algorithm is deterministic, the execution sequence on input $Q = P
  \setminus C$ is the same, and the algorithm returns, again, the trivial
  semigroup consisting of one element. However, $[Q] \ne A^\omega$ and thus,
  the algorithm is incorrect.
\end{proof}

\section{Language operations on morphisms}
\label{sec:langop}

One of the merits of strong recognition is that complementation is easy.
If a morphism $h \colon A^+ \to S$ onto a finite semigroup $S$ strongly
recognizes a language $L \subseteq A^\omega$, it also strongly recognizes the
complement $A^\omega \setminus L$.
As in the case of finite words, we can use \emph{direct
products} for unions and intersections.

Another operation on languages which is of particular interest when it comes to
converting MSO formulas to strongly recognizing morphisms are so-called
\emph{length-preserving morphisms}.
Suppose we are given alphabets $A$, $B$ and a length-preserving morphism $\pi
\colon A^+ \to B^+$, \ie$\pi(a) \in B$ for all $a \in A$. We naturally extend
this morphism to infinite words by setting $\pi(a_1 a_2 \cdots) = \pi(a_1)
\pi(a_2) \cdots$ and to languages $L \subseteq A^\omega$ by setting $\pi(L) =
\set{\pi(\alpha)}{\alpha \in L}$.
\begin{proposition}
  Let $\pi \colon A^+ \to B^+$ be a length-preserving morphism, let $S$ be a
  finite semigroup and let $h \colon A^+ \to S$ be a surjective morphism that
  strongly recognizes a language $L \subseteq A^\omega$.
  Then there exist a semigroup $T$ of size $2^{\abs{S}}$ and a morphism $g
  \colon B^+ \to T$ that strongly recognizes $\pi(L)$.
\end{proposition}
\begin{proof}
  We first define $T$ to be the set $2^S$ of all subsets of $S$ and extend it
  to a semigroup by defining an associative multiplication $X \cdot Y =
  \set{xy}{x \in X, y \in Y}$. The morphism $g \colon B^+ \to T$ is uniquely
  defined by $g(a) = h(\pi^{-1}(a))$ for all $a \in B$.

  Let us now verify that $g$ strongly recognizes $\pi(L)$.
  Consider a linked pair $(s, e)$ and two infinite words $\alpha, \beta \in
  g^{-1}(s)(g^{-1}(e))^\omega$. By Proposition~\ref{prop:strong}, it suffices
  to show that $\alpha \in \pi(L)$ implies $\beta \in \pi(L)$.
  If $\alpha$ is contained in $\pi(L)$, we can conclude by Ramsey's theorem
  that there exists a linked pair $(t, f)$ of $S$ with $t \in s$, $f \in e$ and
  $h^{-1}(t)(h^{-1}(f))^\omega \intersect L \ne \emptyset$. By assumption, $h$
  strongly recognizes $L$ and thus, we have $h^{-1}(t)(h^{-1}(f))^\omega
  \subseteq L$. Since we know that there exists an infinite word $u v_1 v_2
  \cdots \in \pi^{-1}(\beta)$ such that $h(u) = t$ and $h(v_i) = f$ for all $i
  \ge 1$, this immediately yields $u v_1 v_2 \cdots \in L$ and hence $\beta \in
  \pi(L)$.
\end{proof}

\section{Experimental results}

In order to test the algorithms and constructions in practice, we implemented
the conversion of~MSO formulas into strongly recognizing morphisms.
The constructions described in Section~\ref{sec:langop} are used to recursively
convert the formulas, and all intermediate results are minimized using the
algorithm from Section~\ref{sec:syntactic}.
For details on MSO logic over infinite words and its connexion to
regular languages, we refer to~\cite{Thomas90,Thomas97}.
The conversion to strongly recognizing morphisms instead of B\"uchi automata
has the advantage that all intermediate objects can be minimized efficiently. 
Table~\ref{tab:results} shows the size of the computed syntactic semigroup $S$,
the number of linked pairs $F$ and the size of the accepting set $P$ (which is
closed unter conjugation) for the following three families of MSO formulas with
parameter $k \ge 1$ and free second-order variables $X_{k+1} = X_1, X_2, \dots,
X_k$:
\begin{ceqn}
  \begin{align*}
    \varphi_k & \; = \; \forall x \: \bigwedge_{i=1}^k \exists y \: (x < y \land y \in X_i) \\
    \psi_k & \; = \; \forall x \forall y \: (y = x + 1) \rightarrow \bigwedge_{i=1}^k (x \in X_i \rightarrow y \in X_{i+1}) \\
    \chi_k & \; = \; \forall x \bigwedge_{i=1}^k (x \in X_i \rightarrow \exists y \: (x < y \land (y \in X_{i-1} \lor y \in X_{i+1})))
  \end{align*}
\end{ceqn}
All computations were made on a Intel Core i5-3320M with 4GiB of RAM\@. The
execution time was less than three seconds for each formula.
\newcolumntype{R}{>{\raggedleft\arraybackslash}p{1cm}}
\begin{table}[tp!]
  \begin{tabular}{lRRRRRRRRR}
    \toprule
    & \multicolumn{3}{c}{$\varphi_k$} &
      \multicolumn{3}{c}{$\psi_k$} &
      \multicolumn{3}{c}{$\chi_k$} \\
    \cmidrule(r){2-10}
    & $\abs{S}$ & $\abs{F}$ & $\abs{P}$ & $\abs{S}$ & $\abs{F}$ & $\abs{P}$ & $\abs{S}$ & $\abs{F}$ & $\abs{P}$ \\
    \midrule
    $k = 2$ & $4$   & $5$   & $1$ & $12$   & $15$   & $10$   & $7$   & $14$  & $11$   \\
    $k = 3$ & $8$   & $22$  & $1$ & $43$   & $50$   & $41$   & $11$  & $26$  & $15$   \\
    $k = 4$ & $16$  & $74$  & $1$ & $148$  & $163$  & $146$  & $17$  & $61$  & $30$   \\
    $k = 5$ & $32$  & $232$ & $1$ & $539$  & $570$  & $537$  & $41$  & $227$ & $85$   \\
    $k = 6$ & $64$  & $710$ & $1$ & $1863$ & $1926$ & $1861$ & $105$ & $716$ & $184$  \\
    \bottomrule
  \end{tabular}
  \caption{Experimental results for different parameter values}
  \label{tab:results}
\end{table}

\section{Summary and Outlook}

We described several algorithms for weakly recognizing morphisms and
strongly recognizing morphisms over infinite words.
Our tests indicate that strongly recognizing morphisms, when combined with
the minimization algorithm presented in Section~\ref{sec:syntactic}, are a
practical alternative to automata-based models when it comes to deciding
properties of MSO formulas.

Some of the algorithms leave room for optimization. In particular, it would be
interesting to see whether there is a linear-time algorithm to compute
conjugacy classes and whether the running time of the algorithm described in
Section~\ref{sec:strong} can be improved to $\bigO{\abs{A} \cdot \abs{S^2}}$.

{%

\begin{thebibliography}{10}

\bibitem{Arnold85}
A.~Arnold.
\newblock A syntactic congruence for rational $\omega$-languages.
\newblock {\em Theoretical Comput. Sci.}, 39:333--335, 1985.

\bibitem{CalbrixEtAl1994}
H.~Calbrix, M.~Nivat, and A.~Podelski.
\newblock Ultimately periodic words of rational $\omega$-languages.
\newblock In {\em MFCS 94, Proceedings}, volume 802 of {\em LNCS}, pages
  554--566. Springer, 1994.

\bibitem{DiekertGastin08}
V.~Diekert and P.~Gastin.
\newblock First-order definable languages.
\newblock In J.~Flum, E.~Gr{\"a}del, and T.~Wilke, editors, {\em Logic and
  Automata: History and Perspectives}, Texts in Logic and Games, pages
  261--306. Amsterdam University Press, 2008.

\bibitem{FroidurePin97}
V.~Froidure and J.-E. Pin.
\newblock Algorithms for computing finite semigroups.
\newblock In F.~Cucker and M.~Shub, editors, {\em Foundations of Computational
  Mathematics}, pages 112--126. Springer, 1997.

\bibitem{Hopcroft71}
J.~Hopcroft.
\newblock An $n \log n$ algorithm for minimizing states in a finite automaton.
\newblock In Z.~Kohavi and A.~Paz, editors, {\em Theory of Machines and
  Computations}, pages 189--196. Academic Press, New York, 1971.

\bibitem{HopcroftKarp71}
J.~Hopcroft and R.~Karp.
\newblock A linear algorithm for testing equivalence of finite automata.
\newblock Technical report, Dept. of Computer Science, Cornell Univ., December
  1971.

\bibitem{MeyerStockmeyer72}
A.~R. Meyer and L.~J. Stockmeyer.
\newblock The equivalence problem for regular expressions with squaring
  requires exponential space.
\newblock In {\em 13th Annual Symposium on Switching and Automata Theory},
  pages 125--129. IEEE Computer Society, 1972.

\bibitem{PerrinPin04}
D.~Perrin and J.-{\'E}. Pin.
\newblock {\em Infinite words}, volume 141 of {\em Pure and Applied
  Mathematics}.
\newblock Elsevier, 2004.

\bibitem{Pecuchet86}
J.-P. Pécuchet.
\newblock Variet{\'e}s de semisgroupes et mots infinis.
\newblock In {\em STACS 86}, volume 210 of {\em LNCS}, pages 180--191.
  Springer, 1986.

\bibitem{rs59}
M.~O. Rabin and D.~Scott.
\newblock Finite automata and their decision problems.
\newblock {\em IBM Journal of Research and Development}, 3:114--125, 1959.
\newblock Reprinted in E.~F.~Moore, editor, {\em Sequential Machines: Selected
  Papers}, Addison-Wesley, 1964.

\bibitem{SistlaVardiWolper87}
A.~P. Sistla, M.~Y. Vardi, and P.~L. Wolper.
\newblock The complementation problem for {B}{\"u}chi automata with
  applications to temporal logic.
\newblock {\em Theoretical Comput. Sci.}, 49(2-3):217--237, 1987.

\bibitem{Sto74}
L.~J. Stockmeyer.
\newblock The complexity of decision problems in automata theory and logic.
\newblock {PhD} thesis, {TR} 133, M.I.T., Cambridge, 1974.

\bibitem{Tarjan75}
R.~E. Tarjan.
\newblock Efficiency of a good but not linear set union algorithm.
\newblock {\em J. ACM}, 22(2):215--225, Apr. 1975.

\bibitem{Thomas90}
W.~Thomas.
\newblock Automata on infinite objects.
\newblock In J.~van Leeuwen, editor, {\em Handbook of Theoretical Computer
  Science}, chapter~4, pages 133--191. Elsevier, 1990.

\bibitem{Thomas97}
W.~Thomas.
\newblock Languages, automata and logic.
\newblock In A.~Salomaa and G.~Rozenberg, editors, {\em Handbook of Formal
  Languages}, volume 3, Beyond Words, pages 389--455. Springer, Berlin, 1997.

\end{thebibliography}

}

\end{document}